\documentclass[11pt]{article}
\usepackage{amsmath,amssymb,amsfonts} 
\usepackage{epsfig} \usepackage{latexsym,nicefrac,bbm}
\usepackage{xspace}
\usepackage{color,fancybox,graphicx,subfigure,fullpage}
\usepackage[top=1in, bottom=1in, left=1in, right=1in]{geometry}
\usepackage{tabularx} \usepackage{hyperref} 
\usepackage{pdfsync}
\usepackage[boxruled]{algorithm2e}
\usepackage{multicol}
\newcommand{\parag}[1]{ {\bf \noindent #1}}
\newcommand{\defeq}{\stackrel{\textup{def}}{=}}
\newcommand{\nfrac}{\nicefrac}

\newcommand{\supp}{\mathrm{supp}}

\newcommand{\conv}{\mathrm{conv}}

\newcommand{\cM}{\mathcal{M}}
\newcommand{\cB}{\mathcal{B}}

\newcommand{\capa}{\mathrm{Cap}}
\newcommand{\dcapa}{\underline{\mathrm{Cap}}}
\newcommand{\st}{\mathrm{s.t.}}
\newcommand{\un}{\mathrm{un}}

\newcommand{\Pb}{\mathbb{P}}

\newcommand{\mb}{{M(\cB)}}
\newcommand{\mlb}{{M_{\mathrm{lin}}(\cB)}}

\newcommand{\eps}{\varepsilon}

\newenvironment{proof}{\noindent{\bf Proof:}\hspace*{1em}}{\qed\bigskip}
\newcommand{\qed}{\hfill\ensuremath{\square}}

\def\showauthornotes{0} 
 
\def\showdraftbox{0}

\newtheorem{theorem}{Theorem}[section]

\newtheorem{definition}{Definition}[section]
\newtheorem{lemma}[theorem]{Lemma}
\newtheorem{remark}[theorem]{Remark}

\newtheorem{corollary}{Corollary}[section]

\newtheorem{fact}[theorem]{Fact}

\def\FullBox{\hbox{\vrule width 6pt height 6pt depth 0pt}}

\def\qed{\ifmmode\qquad\FullBox\else{\unskip\nobreak\hfil
\penalty50\hskip1em\null\nobreak\hfil\FullBox
\parfillskip=0pt\finalhyphendemerits=0\endgraf}\fi}

\def\qedsketch{\ifmmode\Box\else{\unskip\nobreak\hfil
\penalty50\hskip1em\null\nobreak\hfil$\Box$
\parfillskip=0pt\finalhyphendemerits=0\endgraf}\fi}

\newenvironment{proofof}[1]{\begin{trivlist} \item {\bf Proof
#1:~~}}
  {\qed\end{trivlist}}

\newcommand\N{\mathbb N}
\newcommand\R{\mathbb R}
\newcommand\C{\mathbb C}

\newcommand{\marginlabel}[1]%
{\mbox{}\marginpar{\it{\raggedleft\hspace{0pt}#1}}}

\definecolor{Mygray}{gray}{0.8}

 \ifcsname ifcommentflag\endcsname\else
  \expandafter\let\csname ifcommentflag\expandafter\endcsname
                  \csname iffalse\endcsname
\fi

\ifnum\showauthornotes=1
\newcommand{\Authornote}[2]{{\sf\small\color{red}{[#1: #2]}}}
\newcommand{\Authoredit}[2]{{\sf\small\color{red}{[#1]}\color{blue}{#2}}}
\newcommand{\Authorfnote}[2]{\footnote{\color{red}{#1: #2}}}
\newcommand{\Authorfixme}[1]{\Authornote{#1}{\textbf{??}}}
\newcommand{\Authormarginmark}[1]{\marginpar{\textcolor{red}{\fbox{
#1:!}}}}
\else
\newcommand{\Authornote}[2]{}
\newcommand{\Authoredit}[2]{}
\newcommand{\Authorcomment}[2]{}
\newcommand{\Authorfnote}[2]{}
\newcommand{\Authorfixme}[1]{}
\newcommand{\Authormarginmark}[1]{}
\fi

\newcommand{\inparen}[1]{\left(#1\right)}             
\newcommand{\inbraces}[1]{\left\{#1\right\}}           
\newcommand{\inangle}[1]{\left\langle#1\right\rangle} 

\newcommand{\diag}[1]{{\sf Diag}\left({#1}\right)}

 {
	\begin{enumerate}}{\end{enumerate}}

\newcounter{lecnum}

\newlength{\tpush}
\ifnum\showdraftbox=1

\else

\fi

\title{{\bf Real Stable Polynomials and Matroids:\\  Optimization and Counting } }

\date{}

\usepackage{authblk}

\author[1]{ Damian Straszak}
\author[2]{Nisheeth K. Vishnoi}
\affil[1,2]{\small \'{E}cole Polytechnique F\'{e}d\'{e}rale de Lausanne (EPFL), Switzerland}

\begin{document}
\maketitle

\thispagestyle{empty}

\begin{abstract}

A great variety of fundamental optimization and counting problems  arising in computer science, mathematics and physics  can be reduced to  one of the following computational tasks involving polynomials and set systems: given an oracle access to an $m$-variate real polynomial $g$ and to a family of (multi-)subsets $\cB$ of $[m]$, (1) find  $S \in \cB$ such that the monomial in $g$ corresponding to $S$ has the largest coefficient in $g$,  or (2) compute the sum of coefficients of monomials in $g$ corresponding to all the sets that appear in $\cB$.
Special cases of these problems, such as  computing permanents and mixed discriminants, sampling from determinantal point processes, and maximizing subdeterminants with combinatorial constraints have been topics of much recent interest in theoretical computer science.

In this paper we present a very general convex programming framework geared to solve both of these problems.
Subsequently, we show that roughly,  when $g$ is a {\em real stable} polynomial with non-negative coefficients and $\cB$ is a matroid, the integrality gap of our convex relaxation is finite and depends only on $m$ (and not on the coefficients of $g$) -- in fact, in most interesting cases it is never worse than $e^m$.

Prior to our work, such results were known only in  important but sporadic cases that relied heavily on the structure of either $g$ or $\cB$; it was not even a priori clear if one could formulate a convex relaxation that has a finite integrality gap beyond these special cases.
Two notable examples are a result by Gurvits \cite{Gurvits06} on the van der Waerden conjecture 
for all real stable $g$ when $\cB$ contains {\em one} element, and a result by Nikolov and Singh \cite{NS16}
for a family of \emph{multilinear} real stable polynomials when $\cB$ is the partition matroid.
Our work, which encapsulates almost all interesting cases of $g$ and $\cB$,  benefits from both -- we were inspired by the latter in coming up with the right convex programming relaxation and the former in deriving the integrality gap.  
However,  proving our results requires significant extensions of both; in that process we come up with new notions and connections between real stable polynomials and matroids which should be of independent and wide interest.

\end{abstract}

\newpage

\thispagestyle{empty}

\tableofcontents
\newpage

\setcounter{page}{1}

\section{Introduction}

Several fundamental problems that arise in computer science, mathematics and physics can be formulated as  the following computational tasks regarding multivariate polynomials:  
given oracle access to an $m$-variate real polynomial  $g(x_1, x_2, \ldots, x_m)=\sum_{\alpha \in \N^m} g_\alpha x^\alpha$, where $x^\alpha$ denotes $\prod_{i=1}^m x_i^{\alpha_i}$, and to  a family $\cB \subseteq \N^m$ of multi-indices:
\vspace{-2mm}
\begin{enumerate}
\item {\bf Optimization.} Can we find $\max_{\alpha \in \cB} g_\alpha$ efficiently?
\vspace{-2mm}
\item {\bf Counting.} Can we compute the sum $g_\cB := \sum_{\alpha \in \cB} g_\alpha$ efficiently?
\end{enumerate}
\vspace{-2mm}
We do not restrict the number of monomials in $g$ or the size of $\cB$ -- they can be exponential; indeed if both are polynomially many then the problem is easy.
Instead, we assume access to an evaluation oracle; given any input $x\in \R^m$ the oracle returns $g(x)$. 
Similarly for $\cB$, we assume that an appropriate {\em separation} oracle is provided. 

Such a setting is very general -- on the one hand it captures most counting and discrete optimization problems,  on the other hand, it is easily to seen to be intractable, even if $\cB$ contains a single element. 
Indeed, if only a small (polynomial) number of input-output pairs for $g$ are known, an adversary has a large flexibility in choosing the coefficients of $g$. 
To escape this problem (at least partly) we assume $g$ has nonnegative coefficients. 
Nonnegative polynomials are already quite general and appear frequently in statistical physics, probability (as generating polynomials of  distributions),  machine learning,  as well as in combinatorics~\cite{PemantleHyperbolic}.
Indeed,  all the polynomials underlying the special cases mentioned in the abstract are nonnegative. %

\noindent
{\bf Permanents and Mixed Discriminants.} One of the most studied combinatorial counting problems is the permanent of a nonnegative square matrix $A\in \R_{\geq 0}^{m\times m}$. 
It is defined as $\mathrm{per}(A)=\sum_{\sigma \in S_m} \prod_{i=1}^m A_{i,\sigma(i)}$. 
This problem is known to be $\#P$-complete, hence no exact algorithm for computing permanents is expected to exist. 
However, interestingly, the counting problem (2) above can be used to express the permanent of $A$ by letting $g(x)=\prod_{i=1}^m(\sum_{j=1}^m x_j A_{i,j})$ and $\cB=\{(1,1,\ldots,1)\}$.
Clearly $g(x)$ has nonnegative coefficients and is easy to evaluate on any input;  still, computing its  multilinear coefficient is a hard problem. 
Mixed discriminants provide a powerful generalization of permanents that appears naturally in a variety of settings, e.g., as coefficients of mixed characteristic polynomials which played an important role in the recent resolution of the Kadison-Singer problem \cite{MSS1, MSS2}. 
They can be captured similarly by problem (2) as they arise as coefficients of determinantal polynomials of the form $\det (\sum_{i=1}^m x_i A_i)$ where $A_1, \ldots, A_m \in \R^{m\times m}$ are positive semidefinite matrices.  
\noindent
{\bf Determinantal Point Processes (DPP).}
A DPP is a probability distribution over subsets of $[m]=\{1,2,\ldots, m\}$ defined with respect to a positive semidefinite matrix $L\in \R^{m\times m}$ such that for all $S\subseteq [m]$ we have $\Pb(S)\propto \det(L_{S,S})$, where $L_{S,S}$ is the principal minor of $L$ corresponding to columns and rows in $S$. 
In fact, if the vectors $v_1, \ldots, v_m \in \R^n$ come from the Cholesky decomposition of $L$, then the determinantal polynomial $q(x)=\det\inparen{ \sum_{i=1}^m x_i v_i v_i^\top}$ is a generating polynomial for such a distribution. 
DPPs are important objects of study in combinatorics, probability, physics and more recently in computer science, as they provide excellent models for diversity-based sampling in machine learning (\cite{KuleszaTaskar12}). 
The applicability of DPPs to real life problems crucially relies on efficient algorithms for solving computational problems related to them. These include the problem of sampling from a DPP, computing its partition function and the MAP-inference problem which asks to find the set of highest probability (or equivalently to find the largest coefficient of $q(x)$). 
For the case of unconstrained DPPs problems (1) and (2) are quite well understood, and various solutions have been proposed \cite{Khachiyan95, DR10, AGR16, Nikolov15, Summa15}. Recently, the case of constrained DPPs -- when the support is restricted to a combinatorial family $\cB \subseteq 2^{[m]}$ -- has been studied \cite{NS16,SV16,CDKV16} with machine learning applications in mind, however, very little is known computationally.

\noindent
{\bf Related prior work.} All the examples discussed so far are special cases of the counting and optimization problems mentioned in the beginning. 
Prior results for such problems are far and few -- they either work with a specific class of polynomials or with a simple family $\cB$. 
An important work by Gurvits \cite{Gurvits06} focused on recovering the coefficient of the multilinear term $x_1x_2 \cdots x_m$ from an $m-$variate polynomial $g$. 
In our language, this corresponds to setting $\cB=\{(1,1,\ldots, 1)\}$. 
Towards this, he introduced the notion of \emph{capacity} of a polynomial $g$.
The capacity of a polynomial can be shown to be efficiently computable (given access to an evaluation oracle of the polynomial).
Gurvits proved that, when the polynomial $g$ is {\em real stable}, its capacity is a good (multiplicative) approximation of the coefficient of interest -- in particular the approximation factor does not depend on the coefficients of $g$. 
Real stability is a geometric condition on the set of zeros of a polynomial
and holds for many important classes of polynomials, in particular, for the polynomial $p_A$ and determinantal polynomials introduced earlier.
In fact, Gurvits derived an approximation bound of $e^m$ when $g$ is an $m-$homogeneous real stable polynomial with nonnegative coefficients.\footnote{A polynomial $g(x) = \sum_{\alpha} g_\alpha x^\alpha$ is $n-$homogeneous if for every $\alpha$ with $g_\alpha \neq 0$ we have $\sum_{i=1}^m \alpha_i=n$.} 

The case of a determinantal polynomial $g(x) = \det\inparen{\sum_{i=1}^m x_i v_i v_i ^\top }$ and $\cB$ a set of bases of a partition matroid 
was recently studied by Nikolov and Singh \cite{NS16} where
they presented a polynomial time $e^n$-approximate algorithm to estimate the value of the optimization problem $\max_{S\in \cB} g_S$ (where $n$ is the rank of the matroid). 
For general homogeneous real stable polynomials $g$ and $\cB$ of the form ${[m] \choose n}$,\footnote{We identify sets $S\subseteq [m]$ with their characteristic vectors $1_S \in \N^m$.} an $e^n$-approximation algorithm for the counting problem $\sum_{S\in \cB} g_S$ was obtained in~\cite{AOSS16}. 
Exact counting algorithms for determinantal polynomials were obtained in~\cite{SV16} under the condition that $\cB$ has a description of {\em constant dimension}.

Except the latter (which addresses the case of exact computation), the results so far relied on coming up with capacity-like quantities that can be computed efficiently using \emph{convex programming}. 
Then, using the properties of real stable polynomials, they were shown to approximate the quantities of interest. 
The question of whether such an approach would work more generally for matroids was left open. 
In fact, one of the key difficulties to extend this approach beyond partition matroids was to come up with a notion of capacity that can be captured by convex programming and is zero if the quantity it is trying to approximate is zero.

\noindent
{\bf Our contribution.} In this paper we introduce a new  notion, the $\cB-$capacity of a polynomial $g(x_1, \ldots, x_m)$ with respect to a family $\cB \subseteq 2^{[m]}$, denoted by $\capa_\cB(g)$. 
It enjoys good computational properties, as  it can be evaluated given oracle access to $g$ and to  $\cB$. 
The main question then is: can $\capa_\cB(g)$ serve as a good approximation to $g_\cB$?

We prove that, assuming $g$ is nonnegative, homogeneous and real stable, the $\cB-$capacity of $g$ approximates $g_\cB$ within a factor that depends only on $m$ whenever $\cB$ arises as a set of bases of a matroid $\cM$ (under mild conditions on $\cM$).
 Surprisingly, we also observe that when either of these conditions fails, one cannot hope for this result to hold. 
 %
%

As a consequence we obtain polynomial time algorithms that estimate the sum $\sum_{S\in \cB} g_S$ for any nonnegative real stable polynomial $g$ and a large class of matroid families $\cB$ up to an approximation factor no worse than $e^m$;  this factor can be improved if we know more about the structure of $g$ or $\cB$.
Further, using our notion of $\cB-$capacity, we are able to design a convex program for approximating $\max_{S\in \cB} g_S$. 
 We show that the approximation ratio is essentially of the same order as the guarantee achieved when estimating $g_\cB$ by $\capa_\cB(g)$. 
 This gives one  common framework under which all previous results can be understood and treated as special cases and, as an aside, provides another interesting connection between real stable polynomials and matroids~\cite{COSW04, Branden07}.
Moreover, it provides non-trivial approximation algorithms for various important open problems, such as the DPP MAP-inference problem under matroid constraints. 

\subsection{Statement and Overview of Our Results}

We consider real multivariate polynomials $g$ with nonnegative coefficients, thus $g\in \R^+[z_1, \ldots, z_m]$. 
We are also given a family of sets $\cB \subseteq {2^{[m]}}$. 
We assume that $g$ is given as an evaluation oracle and $\cB$ is given as a separation oracle for the convex hull of $\cB$, i.e., for $P(\cB) = \conv \{ 1_S: S\in \cB\} \subseteq \R^m$. 
Every set $S\subseteq [m]$ gives rise to a  monomial $x^S = \prod_{i\in S} x_i$, by $g_S$ we denote the coefficient of $x^S$ in $g$. 
We consider two computational problems: 
\begin{enumerate} 
\item  finding $\max_{S\in \cB} g_S$ (the optimization problem),
\item computing $\sum_{S\in \cB} g_S$ (the counting problem).
\end{enumerate}

\noindent
We present a general framework for solving such counting and optimization problems. 
We start with the counting problem and later extend our solution to the optimization problem. 
The key idea in our approach is to approximate $g_\cB$ by the optimal value of a convex program defined with respect to $g$ and $\cB$. 
In its rudimentary form, this approach was explicitly pioneered by \cite{Gurvits06}, where the following notion of capacity of an $m-$variate polynomial $g\in \R^+[x_1, \ldots, x_m]$ was introduced:
$$\capa(g)=\inf \left\{g(z): z>0, \prod_{i=1}^m z_i=1 \right\}.$$
It is not hard to see that after introducing new variables $y_i = \log z_i$ and replacing the objective by $\log g(z)$, one obtains a convex program, which can be solved efficiently. 

The crucial fact proved by~\cite{Gurvits06} is that whenever $g$ is a real stable \emph{and} homogeneous polynomial, then $\capa(g)$ approximates $g_{[m]}$ (the coefficient of $\prod_{i=1}^m x_i$ in $g(x)$) up to a multiplicative factor of $e^m$, i.e., the approximation guarantee does not depend on the coefficients $p$. 
It is important to note that no such result holds  if we put no restrictions on $g$.\footnote{Indeed for $g(x)=\sum_{i=1}^m x_i^m + \eps \prod_{i=1}^m x_i$ one can easily see that $\capa(g)\geq 1$, while $g_{[m]}=\eps$ can be arbitrarily small. }

A polynomial $g\in \R[x_1, \ldots, x_m]$ is said to be real stable if none of its roots $x\in \R^m$ satisfy $\Im(x_i)>0$, for every $i=1,2,\ldots, m$. For univariate polynomials, real stability is equivalent to real-rootedness, which generally speaking implies good analytic properties of such polynomials. 
Real stable polynomials have been extensively studied in mathematics, in particular a complete characterization of their {\em closure} properties under linear maps is known~\cite{BB09, BB09b,Wagner11}. 
Many natural polynomials (such as the determinantal polynomial $\det(\sum_{i} x_i A_i)$) are known to be real stable, and many other can be derived by applying these closure properties to them. 
Real stable polynomials have recently found many applications in combinatorics, probability and theoretical computer science (\cite{BBL09,PemantleHyperbolic,MSS1,MSS2,Vishnoi-Survey,ShayanFOCS2015,NS16, AOSS16,AGR16}).

In an attempt to provide an efficiently computable estimate for $\sum_{S\in \cB} g_S$ we propose the following new notion of $\cB-$capacity of a polynomial. 
\begin{definition}[$\cB-$capacity]
For a polynomial $g \in \R^+[x_1, \ldots, x_m]$ and any family of sets $\cB \subseteq 2^{[m]}$ we define the $\cB-$capacity of $g$ to be
$$\capa_{\cB}(g) = \sup_{\theta \in P(\cB)} \inf_{z>0} \frac{g(z)}{\prod_{i=1}^m z_i^{\theta_i}}.$$
\end{definition}

\noindent
When $n=m$ and $\cB=\{\{1,2,\ldots,m\}\}$ one recovers Gurvits' notion of capacity from ours.
  In~ Section \ref{ssec:computability} we prove that $\capa_\cB(g)$ can be stated as a convex program and thus computed efficiently.
Additionally, when the coefficients of $g$ correspond to a probability distribution, our $\capa_\cB(g)$ has a natural interpretation in terms of entropy~\cite{SinghV14}, see Section~\ref{sec:entropy}.
Importantly, an (equivalent) dual characterization of $\cB-$capacity allows us to prove that $\capa_\cB(g)$ is an upper bound on  $g_\cB$, see Section~\ref{ssec:dual}. 
This duality for the special case  $\cB={[m] \choose n}$ was also observed in~\cite{AOSS16}.
Our next definition captures how good an approximation $\capa_\cB(g)$ is to $g_\cB$.
We define the following two approximation ratios: the second one being for the case when $g$ is a multilinear polynomial.
We provide both as, in applications, the polynomial $g$ is often multilinear and stronger bounds can be obtained for this setting.
\begin{definition}[Approximation Ratios]
For a family $\cB \subseteq {[m] \choose n}$ we define:
\begin{align*}
\mb& \defeq \sup \inbraces{ \frac{\capa_\cB(p)}{p_\cB}: p\in \R^+ [x_1, \ldots, x_m] \mbox{ is real stable, $n$-homogeneous and $p_\cB>0$}},\\
\mlb& \defeq \sup \inbraces{ \frac{\capa_\cB(p)}{p_\cB}: p\in \R_1^{+} [x_1, \ldots, x_m] \mbox{ is real stable, $n$-homogeneous and $p_\cB>0$}}.\\
\end{align*}
\end{definition}
Here, $\R^+_{1}[x_1, \ldots, x_m]$ denotes the set of all multiaffine polynomials with nonnegative coefficients. 
It follows from these definition that $\mlb \leq \mb$.
The first question one can ask is: are these quantities finite?
Our first result gives a sufficient condition for $\mb$ being bounded and makes a connection between real stable polynomials and matroids.
In particular, it relies on the interplay between matroids and supports of {\em Strongly Rayleigh} distributions.\footnote{A distribution $\mu$ over subsets of $[m]$ is called Strongly Rayleigh if its generating polynomial $g(z)= \sum_{S\subseteq [m]} \mu(S) z^S$ is real stable.} 

\begin{theorem}[Finiteness of $\mb$]\label{thm:finite}
Let $\cB \subseteq {[m] \choose n}$ be a family of bases of a matroid and let $\cB^\star$ be the family of bases of the dual matroid. If there exists a strongly Rayleigh distribution supported on $\cB^\star$, then $\mb < \infty$.
\end{theorem}
Strongly Rayleigh distributions have been extensively studied~\cite{COSW04,Branden07,BBL09}, and it is known that most\footnote{A non-example has been discovered by \cite{BBL09}: the 7-element Fano matroid.} matroids satisfy the condition stated in Theorem~\ref{thm:finite}. 
Interestingly, when one gives up either real stability or the assumption that $\cB$ comes from a matroid, then $\mb$ can be infinite, 
we provide such examples in Section~\ref{sec:examples}.

The proof of Theorem \ref{thm:finite} appears in Section~\ref{ssec:proofs}, below we describe the key steps.
We first connect our notion of capacity to that of Gurvits via an inequality of the form (see Lemma~\ref{lemma:capa})\footnote{A similar inequality, for a specific choice of $h$ and $\cB$ was considered in~\cite{AOSS16}.}
\begin{equation}\label{eq:capa_ineq}
\capa(g\cdot h) \geq \capa_{\cB}(g) \cdot \dcapa_{\cB^\star} (h),
\end{equation}
where $\dcapa_{\cB^\star}$ denotes what we call the {\em lower capacity} of  $\cB^\star$ -- it is defined as
$$\dcapa_{\cB^\star}(g) = \inf_{\theta \in P(\cB^\star )} \inf_{z>0} \frac{g(z)}{\prod_{i=1}^m z_i^{\theta_i}}.$$
Note that in the definition of lower capacity the supremum is replaced by infimum.
Since inequality~\eqref{eq:capa_ineq} holds for every polynomial $h$, it allows us to upper-bound the $\cB-$capacity by providing an appropriate $h$. 
Our choice of $h$ should ideally allow us to relate $\capa(g\cdot h)$ to $g_\cB$, as our primary goal is to upper bound the ratio $\frac{\capa_{\cB}(g)}{g_\cB}$. 
To this end we introduce a notion of $\cB^\star-$selection, which essentially describes sufficient conditions on $h$ for this to succeed. 
We require that the coefficient of the monomial $\prod_{i=1}^m x_i$ in $(g\cdot h)$ is (exactly or approximately) equal to $p_\cB$. 
Using Gurvits' inequality we prove that for any real stable $\cB^\star-$selection $h$ it holds that 
$$ g_\cB \geq \frac{ m!}{ m^m} \dcapa_{\cB^\star}(h)  \cdot \capa_{\cB}(g). $$
Our task is then reduced to coming up with a good $\cB^\star-$selection $h$, whose lower capacity $\dcapa_{\cB^\star}$ is as large as possible. 
Subsections~\ref{sec:selection1}, \ref{sec:selection2}, \ref{ssec:partition} deal with this problem and propose several choices, depending on $\cB$. 
The most canonical one is 
$$h(z) = \sum_{S\in \cB^\star} z^S,$$
in Lemma~\ref{lemma:dual_capa} we prove that $\dcapa_{\cB^\star}(h)\geq 1$, however one of the conditions we impose on a $\cB^\star-$selection is real stability, hence this particular $h$ gives us the answer only in the case when $\cB^\star$ is a strongly Rayleigh matroid~\cite{Branden07}. 
If this particular ``uniform'' generating polynomial is not real stable, one might consider other polynomials supported on $\cB^\star$. 
For instance, if $\cB$ is a linear matroid then $\cB^\star$ is also linear and there exists a determinantal polynomial $h(z)=\det(\sum_{i=1}^m z_i v_i v_i^\top)$ whose support is exactly $\cB^\star$. 
This leads to the notion of approximate $\cB^\star-$selections, which in the end allows us to prove finiteness of $\mb$ whenever there exists a nonnegative real stable polynomial supported on $\cB^\star$. 
The price we pay to allow this ``non-uniformity'' is an additional term in the approximation guarantee which is equal to the ratio between the largest and the smallest coefficient of $h$ in $\cB^\star$.

Now we present  quantitative bounds on $\mlb$ and $\mb$ for a large class of matroids.
In a nutshell, we can prove that for most interesting matroids, whenever these quantities are finite, they are never worse than $e^m$.
Prior to our work, two bounds of this type, for special cases of matroids, were known. 
For the case of uniform matroids, i.e., $\cB = {[m] \choose n}$, the inequality of~\cite{Gurvits06} yields   $\mb \leq e^n$ for $n=m$, which was recently extended to any $m\geq n$ by~\cite{AOSS16}.\footnote{The paper of Barvinok  \cite{Barvinok12} also uses a capacity like quantity to approximately count $0$-$1$ matrices with prescribed row and column sums.}
For partition matroids it was implicitly shown in~\cite{NS16} that $\mlb \leq e^n$, where $n$ is the rank of the matroid; we state the precise bound in the theorem below and derive it in Section~\ref{app:partition}.
We can recover all these bounds  by using the structure of $g$ or $\cB$. 
For precise definitions of partition matroids, linear matroids and the {\em unbalance} $\un(\cM)$ of a linear matroid we refer to Section~\ref{sec:preliminaries}.

\begin{theorem}[Quantitative Bounds]\label{thm:bounds}
Let  $\cB \subseteq {[m] \choose n}$ be a family of bases of a matroid $\cM$.
\begin{enumerate}
\item {\bf Strongly Rayleigh.} If there is a full support Strongly Rayleigh distribution $\{p_S\}_{S\in \cB^\star}$ on the set $\cB^\star$ of bases of the dual matroid and $P=\max_{S,T\in \cB^\star} \frac{p_S}{p_T}$ then ${\mb} \leq  P \cdot \frac{m^m}{m!}$ and $\mlb \leq P \cdot 2^m$.
\item {\bf Linear Matroids.} If $\cM$ is $\R-$linear,  then ${\mb} \leq \un(\cM^\star) \cdot \frac{m^m}{m!}$ and $\mlb \leq \un(\cM^\star) \cdot 2^m$.
\item {\bf Partition Matroids.} If $\cM$ is a partition matroid $\{U(P_j, b_j)\}_{j\in [p]}$ then ${\mb}\leq \frac{m^m}{m!} \cdot \prod_{j=1}^p \frac{(|P_j|-b_j)!}{(|P_j|-b_j)^{|P_j|-b_j}}\leq e^{n+p} \inparen{\frac{m}{p}}^{p/2}$ and $\mlb \leq \frac{n^n}{n!} \prod_{j=1}^p \frac{b_j!}{b_j^{b_j}}$.
\end{enumerate}
\end{theorem}
Theorem~\ref{thm:bounds} implies in particular that for a large class of matroids (including graphic matroids and all regular matroids which fall under the case (1) and (2)), 
$\capa_\cB(g)$ gives an $e^m-$approximation to $g_\cB$. 
Since (as demonstrated in Section~\ref{ssec:computability}) computing $\capa_\cB(g)$ is possible in polynomial time, this implies an efficient method to estimate $g_\cB$ -- the counting problem.
We do not attempt to optimize the bounds in Theorem~\ref{thm:bounds}; most likely they can be improved, however an exponential dependence on $n$ is inevitable. 
In Section~\ref{app:lower_bound} we provide an example where $\mb$ can be as large as $e^{\sqrt{m}}$ for partition matroids.

We now turn to the optimization problem. 
The challenge in solving the problem $\max_{S\in \cB} g_S$ lies in restricting the optimization process to sets in $\cB$ only, while not making the algorithm too ``discrete''. 
To explain this intuition, assume that $g$ is a multilinear polynomial. 
In such a case we can compute $g_S$ for any set $S$ by just querying $g$ on $1_S$. 
Thus our algorithm could just query a polynomial number of coefficients of $g$ and try to find the largest one using these ``hints''. 
Unfortunately, one can show that every such approach fails, since there are examples where all the coefficients of $g$ in $\cB$ except one are equal to zero, hence it is not possible to find it using such a ``discrete search''. 
Because of that, one is forced to evaluate $g$ on ``nontrivial'' inputs to gain a more global view on the coefficients in $\cB$. 
When trying to achieve this goal one runs into the problem of estimating the contribution of monomials in $\cB$ to the value of $g$ at a queried point. 
In other words: there is no simple way to tell whether a large value of $g(x)$ at a given point $x$ comes from monomials in $\cB$ or from monomials outside of $\cB$, since typically the structure of $\cB$ is quite complicated.
From a high level view point our algorithm runs a global optimization process which internally uses $\cB-$capacity to estimate the contribution of monomials in $\cB$ at a current point. 
We obtain the following theorem.
\begin{theorem}[Optimization]\label{thm:max}
Let $\cB \subseteq {[m]\choose n}$ be any family of sets and let $g\in \R ^+[x_1, \ldots,x_m]$ be any real stable, $n-$homogeneous polynomial. Given access to a separation oracle for $P(\cB)=\conv\{1_S:S\in \cB\}$ and an evaluation oracle for $g$ there is a polynomial time algorithm which estimates the value of 
$$\max_{S\in \cB} g_S$$ 
up to a factor of $\mb \cdot e^n.$ 
In the case when $ g\in \R_1 ^+[x_1, \ldots,x_m]$, $\mb$ can be replaced by $\mlb$ in the bound above.
Further, $e^n$ above can be replaced by a smaller quantity  $A(\cB) = \max \inbraces{\sum_{S\in \cB} z^S: z\in P(\cB)}$ in both cases.
 
\end{theorem}
This theorem, when combined with our quantitative bounds from Theorem \ref{thm:bounds} imply $g$-independent approximate estimators for the optimization problem in a fairly general setting. 
Recovering a solution which attains such an approximation seems to require near-exact solutions for the counting problem -- even in the special case of Nikolov and Singh \cite{NS16}; see \cite{SV16}.
Before we give a sketch of the proof of this theorem, we state an important corollary  concerning maximizing subdeterminants subject to matroid constraints. 
Since the determinantal polynomials which appear in this setting are multilinear, we can use a bound on $\mlb$ in Theorem \ref{thm:max}. 
\begin{corollary}[Subdeterminant Maximization]\label{cor:subdet}
Let $L\in \R^{m\times m}$ be a PSD matrix and $\cB \subseteq {[m]\choose n}$ be a family of sets.  Given access to a separation oracle for $P(\cB)$ there is a polynomial time algorithm which estimates the value of
$\max_{S\in \cB} \det(L_{S,S})$
up to a factor of $\mlb \cdot A(\cB)$, where $A(\cB) = \max \inbraces{\sum_{S\in \cB} z^S: z\in P(\cB)}\leq e^n$.
\end{corollary}
By combining Theorem~\ref{thm:bounds} and the above corollary we can recover the $e^n-$estimation algorithm for maximizing subdeterminants under partition constraints by~\cite{NS16}\footnote{It is not hard to see that for $\cB$ corresponding to a partition matroid $\{U(P_j,b_j)\}_{j\in [p]}$ we have $A(\cB) \leq \prod_{j=1}^p \frac{b_j^{b_j}}{b_j!},$ so that $\mlb \cdot A(\cB) \leq e^n$.}. We also obtain new results, such as a $2^m \cdot e^{n}$-estimation for maximizing subdeterminants under spanning tree constraints and more generally under constraints defined by any regular matroid. 
The proof is provided in Section~\ref{ssec:maxdet}.

The  proof of Theorem~\ref{thm:max} appears in Section~\ref{ssec:convex_relax}. Below we discuss the key ideas.
We introduce a convex relaxation for the optimization problem $\max_{S\in \cB} g_S$. Perhaps the most natural choice for such a relaxation is
$\sup_{x \in P(\cB)}~ g(x_1, \ldots, x_m).$
While the above works well for the uniform case $\cB = {[m]\choose n}$, for nontrivial families $\cB$, it has an unbounded integrality gap.
Instead,  consider the polynomial $r(x)=g(x_1z_1, \ldots, x_m z_m)$ parametrized by $z>0$. 
We have $r_\cB=\sum_{S\in \cB} z^S g_S$. To avoid the influence of terms outside of $\cB$ one can try to maximize $r_\cB$ over $z\in P(\cB)$. Of course $r_\cB$ is not necessarily efficiently computable. But crucially, we know that $\capa_\cB(r)$ provides a good approximation to it. Hence finally we arrive at the following relaxation

$$\sup_{z \in P(\cB)}~~ \capa_\cB(g(x_1 z_1, \ldots, x_mz_m)).$$

\noindent
It is not a priori clear  whether it can be solved efficiently, since $\capa_\cB$ itself is a result of an optimization problem, we show that it reduces to a concave-convex saddle point problem which can be solved using convex optimization tools. Finally, let us note that the above gives a conceptually simple way of arriving at an analogous relaxation proposed in~\cite{NS16} in the context of subdeterminant maximization.\footnote{In fact the  term corresponding to capacity in~\cite{NS16} is slightly different, which causes their relaxation to have infinite integrality gap for most matroids.}

\subsection{Discussion}
To summarize, in this paper we present a unified convex programming framework and use it  to obtain nontrivial approximation guarantees for counting and optimization problems involving the very general setting of real stable polynomials and matroids.
This significantly extends the works of Gurvits \cite{Gurvits06} and Nikolov and Singh \cite{NS16}.
Our definition of capacity also makes sense in the case when the polynomial is not homogeneous and the family is arbitrary (not necessarily of the form $\cB \subseteq {[m] \choose n}$). We can even extend it to the case of $\cB \subseteq \N^m$ by enforcing $z^\alpha \geq 1$ for $\alpha \in \cB$ in the dual characterization of capacity (see~\ref{ssec:dual}). For this general case, however, we are not aware of any nontrivial sufficient conditions which guarantee that $\capa_\cB(g)$ approximates $g_\cB$ well.
Our bounds on  both $\mb$ and $\mlb$ are unlikely to be optimal. The problems seem to have better {\em structure} for multilinear polynomials and it is an interesting problem to see to what extent our bounds can be improved. 

\subsection{Organization of the Paper}
The remaining part of the paper is structured as follows. It starts with Preliminaries, containing basic definition and necessary background. Section~\ref{sec:counting} is devoted to the proof of Theorem~\ref{thm:finite} and Theorem~\ref{thm:bounds}.  The next Section~\ref{sec:counting} contains the proof of Theorem~\ref{thm:max} and Corollary~\ref{cor:subdet}. In the Section~\ref{sec:entropy} we provide an interpretation of $\cB-$capacity in terms of entropy. Next in Section~\ref{sec:examples} we provide examples when $\capa_\cB(g)$ does not provide a good approximation to $g_\cB$, in case when $g$ is not real stable or $\cB$ is not a family of bases of a matroid. Section~\ref{app:lower_bound} provides a lower bound on $\mb$ for the partition case. Finally in Section~\ref{app:partition} we derive a bound on $\mlb$ for partition matroids.

\section{Preliminaries}\label{sec:preliminaries}

\parag{Multivariate Polynomials. } We consider real polynomials in $m$ variables: $\R[x_1, \ldots, x_m]$. Every such polynomial can be written as a finite sum $g(x) = \sum_{\alpha \in \N^m} g_\alpha x^\alpha$, where $g_\alpha \in \R$ and $x^\alpha$ denotes the monomial $\prod_{i=1}^m x_i^{\alpha_i}$. The number $g_\alpha$ we call the coefficient of $x^\alpha$ in $g$. The degree of a monomial $x^\alpha$ is defined as $|\alpha|\defeq\sum_{i=1}^m \alpha_i$. We say that $p$ is $n-$homogeneous if $g_\alpha$ is nonzero only for degree $n$ monomials. For a set $S\subseteq [m]$ we often identify the multi-index $1_S$ (the characteristic vector of $S$) with simply $S$, i.e. $z^S = \prod_{i \in S}z_i$. Monomials of this form we call square-free monomials, while the remaining ones we call square monomials. A multiaffine (or square-free) polynomial is one which does contain only square-free monomials. The set of multiaffine polynomials is denoted by  $\R_1[x_1, \ldots, x_m]$. The set of polynomials with nonnegative coefficients is denoted by $\R^+[x_1, \ldots, x_m]$.

\vspace{2mm}
\parag{Real Stable Polynomials and Strongly Rayleigh Distributions. } A polynomial $g\in \R[x_1, \ldots, x_m]$ is called real stable if none of its roots $z=(z_1, \ldots, z_m) \in \C^m$ satisfies: $\Im(z_i)>0$ for every $i=1,2,\ldots, m$. An equivalent characterization is: $g\in \R[x_1, \ldots, x_m]$ is real stable if and only if for all vectors $u\in \R^m$ and $v\in \R_{> 0}^m$ the polynomial $f(t) = g(u_1+v_1t, \ldots, u_m+v_m t)$ is real-rooted. Real stable polynomials are closed under taking partial derivatives and under multiplication. A distribution $\mu$ over subsets of $[m]=\{1,2,\ldots, m\}$ is called Strongly Rayleigh if the polynomial $g(z) = \sum_{S\subseteq [m]} \mu(S) z^S$ is real stable. 

\vspace{2mm}
\parag{Matroids. } For a comprehensive treatment of matroid theory we refer the reader to~\cite{Oxley06}. Below we state the most important definitions and  examples of matroids, which are most relevant to our results. A matroid is a pair $(U, \cM)$ such that $U$ is a finite set and $\cM \subseteq 2^{U}$ satisfies the following three axioms: (1) $\emptyset \in \cM$, (2) if $S\in \cM$ and $S'\subseteq S$ then $S'\in \cM$, (3) if $A,B\in \cM$ and $|A|>|B|$, then there exists an element $a\in A\setminus B$ such that $B\cup \{a\} \in \cM$. The collection $\cB \subseteq \cM$ of all inclusion-wise maximal elements of $\cM$ is called the set of bases of the matroid. It is known that all the sets in $\cB$ have the same cardinality, which is called the rank of the matroid. 
	
Given a matroid $\cM \subseteq 2^{U}$ with a set of bases $\cB$ we define another collection of sets $\cB^\star \subseteq 2^{U}$ to be  $\cB^\star := \{ U\setminus S: S\in \cB\}$. Then $\cB^\star$ can be shown to be a collection of bases of another matroid $\cM^\star$, called the dual of $\cM$. 

\vspace{2mm}
\parag{Linear and Strongly Rayleigh Matroids. } We say that a matroid $\cM \subseteq 2^{[m]}$ is $\R-$linear if there exists a matrix $V\in \R^{m\times n}$ (with rows $v_1, v_2, \ldots, v_m \in \R^n$) such that for every set $S\subseteq [m]$ we have $S\in \cM$ if and only if the collection of vectors $\{v_j: j\in S\}$ is linearly independent over $\R$. Such a matrix $V$ we call an $\R-$representation of the matroid $\cM$. If $\cB$ is a set of bases of $\cM$ and $V \in \R^{m\times n}$ is a representation, we define the unbalance of $V$ to be $\un(V) = \max \inbraces{\frac{\det(V_S^\top V_S)}{\det(V_T^\top V_T)}: S,T\in \cB}$, where $V_S$ is an $|S| \times n$ submatrix of $V$ corresponding to rows from $S$. For an $\R-$linear matroid $\cM$ with set of bases $\cB$ we define $\un(\cM)$ (or equivalently $\un(\cB)$) to be the minimum $\un(V)$ over all $\R-$representations $V$ of this matroid. 

Matroids $\cM$ which have a totally unimodular $\R$-representation are called regular, in such a case $\un(\cM)=1$. 

A matroid $\cM$ with a set of bases $\cB$ is called strongly Rayleigh if the polynomial $g(z)=\sum_{S\in \cB} z^S$ is real stable. Regular matroids are examples of strongly Rayleigh matroids.

\vspace{2mm}
\parag{Examples of Matroids. } If $U=[m]$ and $n\leq m$ then the collection of sets $\cB = {[m] \choose n} = \{S \subseteq [m]: |S|=n\}$ is a set of bases of the so called uniform matroid, which we denote by $U(m,n)$.  If $U=[m]$ and a partition of $[m]$ into non-empty subsets $P_1, P_2, \ldots, P_p$ is given together with numbers $b_1, b_2, \ldots, b_p$, then the collection of sets $\cB=\{S: |S\cap P_j|=b_j \mbox{ for all } j=1,2, \ldots, p\}$ is a family of bases of a partition matroid, which we denote by $\{U(P_j, b_j)\}_{j\in [p]}$. 

If $G$ is an undirected graph with edges labeled by $[m]$, then we can define its graphic (or spanning tree) matroid as follows: the universe is $[m]$ and a set $S\subseteq [m]$ is a basis if and only if $S$ corresponds to a spanning tree in $G$. Graphic matroids are regular.

\section{Counting}\label{sec:counting}
\subsection{An Inequality on $\cB-$Capacity}\label{sec:capa}
In this section we first propose a generalization of the notion of capacity, which was initially introduced and studied  by Gurvits~\cite{Gurvits06}. 

\begin{definition}
Consider  an $m-$variate polynomial  $g\in \R^+[z_1,\ldots, z_m].$ Let $\cB \subseteq {[m] \choose n}$ be any family of sets. We define the $\cB-$capacity of $g$ to be:
$$\capa_{\cB}(g)=\sup_{\theta \in P(\cB)} \inf_{z>0} \frac{g(z)}{\prod_{i=1}^m z_i^{\theta_i}},$$
where $P(\cB)=\conv\{ 1_S: S\in \cB\} \subseteq [0,1]^m$. The lower $\cB-$capacity of $g$ is:
$$\dcapa_{\cB}(g) = \inf_{\theta \in P(\cB)} \inf_{z>0} \frac{g(z)}{\prod_{i=1}^m z_i^{\theta_i}}.$$
\end{definition}
 It is not hard to see that for the setting when $m=n$ and $\cB=\{\{1,2,\ldots, m\} \}$ one recovers the familiar capacity defined in~\cite{Gurvits06}. The main goal of this section is to provide an extension of Gurvits' result which asserts that $\capa(g)\leq \frac{m^m}{m!}g_{[m]}$ (where $g_{[m]}$ is the coefficient of $z^{[m]}$ in $g$) under the assumption that $g$ is $m-$homogeneous, real stable and has nonnegative coefficients. 
 
One of the crucial ingredients of our extension of~\cite{Gurvits06} is the following inequality which ties together the classical capacity and the ones we introduced here.

\begin{lemma}\label{lemma:capa}
Let $g,h\in \R^+[z_1, \ldots, z_m]$ be $m-$variate polynomials, $\cB \subseteq {[m] \choose n}$ be any family of $n-$subsets of $[m]$ and $\cB^\star=\{[m]\setminus S: S\in \cB\}$ be its dual. Then
$$\capa(g\cdot h) \geq \capa_{\cB}(g) \cdot \dcapa_{\cB^\star} (h).$$
\end{lemma}
\begin{proof}
We have, for every $\theta \in [0,1]^m$ that
\begin{equation}\label{eq:inf_inf}
\capa(g\cdot h)=\inf_{z>0} \frac{g(z)\cdot h(z)}{\prod_{i=1}^m z_i} \geq \inf_{z>0} \frac{g(z)}{\prod_{i=1}^m z_i^{\theta_i}} \cdot  \inf_{z>0} \frac{h(z)}{\prod_{i=1}^m z_i^{1-\theta_i}}.
\end{equation}
Note now that whenever $\theta \in P(\cB)$ then $(1-\theta) \in P(\cB^\star)$. To prove it, let $\theta = \sum_{S\in \cB} \alpha_S 1_S$ for some $\{\alpha_S\}_{S\in \cB}$ with $\alpha\geq 0$ and $\sum_{S}\alpha_S=1$. Then:
$$(1-\theta)=\sum_{S\in \cB} \alpha_S(1-1_S) = \sum_{S\in \cB}\alpha_S 1_{\bar{S}}=\sum_{S\in \cB^\star }\alpha_{\bar{S}} 1_S\in P(\cB^\star).$$
By minimizing the second factor in the RHS of~\eqref{eq:inf_inf} over $\theta \in P(\cB)$ we obtain the following:
$$\capa(g\cdot h) \geq \inf_{z>0} \frac{g(z)}{\prod_{i=1}^m z_i^{\theta_i}} \cdot \inf_{\tau \in P(\cB^\star)} \inf _{z>0} \frac{h(z)}{\prod_{i=1}^m z_i^{\tau_i}},$$
for every fixed $\theta \in P(\cB)$. By maximizing the RHS of the above with respect to $\theta$, we finally arrive at:
$$\capa(g\cdot h) \geq \capa_{\cB}(g) \cdot \dcapa_{\cB^\star} (h).$$
\end{proof}

\noindent Since we are interested in proving an upper bound on the $\cB-$capacity we will apply the Lemma~\ref{lemma:capa} in the following way:
$$\capa_{\cB}(g) \leq  \frac{\capa(g\cdot h)}{\dcapa_{\cB^\star} (h)}.$$
Thus the task of upper bounding the capacity boils down to finding an appropriate polynomial $h$, which allows us to relate $\capa(g\cdot h)$ in the RHS to the sum of coefficients $g_{\cB}$. There is some freedom in the choice of $h$, hence one can set the second goal to make the lower capacity $\dcapa_{\cB^\star}$ as large as possible.
 
 Below we provide a definition which makes precise which properties of $h$ are relevant.

\begin{definition}
Let $\cB \subseteq {[m] \choose n}$ and let $h\in \R^+[z_1, \ldots, z_m]$ be an $n-$homogeneous polynomial. We call $h$ a $\cB$-selection if it satisfies the following two conditions
\begin{enumerate}
\item For $S\in {[m] \choose n}$, $h_S>0$ only if $S\in \cB$,
\item $h_{S}=1$ for every $S\in \cB$.
\end{enumerate}
We say that $h$ is a $c-$approximate $\cB$-selection if it satisfies (1) and (2) is replaced by $h_S \in [1,c]$ for every $S\in \cB$ (here $c\geq 1$ is any number).
\end{definition}
Note that $h$ is not assumed to be squarefree, and hence importantly there is quite a lot of flexibility in the choice of a $\cB$-selection.
We are now ready to state and prove the main technical result of this section, which relates $\capa_\cB(p)$ to $p_\cB$, the precision of this approximation depends on the quality of a $\cB^\star$-selection (its lower capacity) one can provide. 

\begin{lemma}\label{lemma:approx}
Let $g\in \R^+[z_1, \ldots, z_m]$ be a real stable $n-$homogeneous polynomial and $\cB \subseteq {[m] \choose n}$ be any family of sets. For every real stable $\cB^\star$-selection $h\in \R^+[z_1, \ldots, z_m]$ we have:
$$\capa_{\cB}(g) \cdot \dcapa_{\cB^\star}(h)\cdot  \frac{ \cdot m!}{ m^m}  \leq  g_\cB  \leq \capa_{\cB}(g).$$
Moreover, if $g,h\in \R_{1}^+[z_1, \ldots, z_m]$, then the term $\frac{m!}{m^m}$ in the bound above can be replaced by $2^{-m}$.
\end{lemma}
\begin{proof}
Consider the polynomial $g(z)\cdot h(z)$ which is real stable, as a product of real stable polynomials. Note that from our assumptions it is homogeneous of degree $m$. We apply Gurvits' inequality~\cite{Gurvits06} to $g\cdot h.$ Let $s\in \R$ be the coefficient of $\prod_{i=1}^m z_i$ in $g\cdot h$, then
$$\capa(g\cdot h) \leq \frac{m^m}{m!} s.$$
Since $h$ is a $\cB^\star -$selection, it is not hard to see that:
$$s=\sum_{S\in \cB} g_S = g_\cB,$$
because the only pairs of monomials from $g$ and $h$ which contribute to $s$ are of the form $x^S$ and $x^{[m]\setminus S}$. By combining this with Lemma~\ref{lemma:capa}, we obtain
$$\capa_{\cB}(g)\cdot \dcapa_{\cB^\star}(h)\leq \capa(g \cdot h) \leq g_\cB\frac{ m!}{ m^m}.$$

To obtain the improved bound under the assumption that $g,h\in \R_1^+[z_1, \ldots, z_m]$ we observe that in the reasoning above, the degree of every variable in the polynomial $g\cdot h$ is at most $2$. Hence we can apply a stronger form of Gurvits' inequality~\cite{Gurvits06}, where $\frac{m^m}{m!}$ is replaced by $\prod_{i=1}^m \inparen{\frac{d_i}{d_i-1}}^{d_i-1}\leq 2^m$, where $d_i$ is the degree of $z_i$ in $g\cdot h$.

This provides us with the LHS of the inequality. The right hand side follows easily from the dual characterization of $\cB-$capacity provided in Lemma~\ref{lemma:dual_capa}.
\end{proof}

\begin{remark}\label{remark:approx_selec}
The above lemma, still holds (with the same proof) when we assume $h$ to be only a $c$-approximate $\cB^\star$-selection. The only required modification is to divide the LHS of the lower bound inequality by $c$. 
\end{remark}

\subsection{Dual characterization of $\cB-$Capacity}\label{ssec:dual}
The way $\capa_\cB(g)$ is defined makes it well suited for proving lower bounds on $g_\cB$. In the following lemma we provide an equivalent, dual characterization, which gives a straightforward upper bound and is often preferred from the computational viewpoint. For the case when $\cB = {[m] \choose n}$ this was observed by~\cite{AOSS16}.

\begin{lemma}[Equivalent definition of $\cB-$capacity]\label{lemma:dual_capa}
Let $g\in \R^+[z_1, \ldots, z_m]$ be an $n-$homogeneous polynomial and $\cB \subseteq {[m] \choose n}$ be any family of sets. Then:
$$\capa_\cB(g) = \inf \{g(z): z>0, z^S \geq 1 \mbox{ for all }S\in \cB\}.$$
\end{lemma}
\begin{proof}
We depart from the following convex program:
\begin{equation}
\begin{aligned}
	\inf_{y\in \R^m}  ~~ &  \log g(e^y),&\\
	\st ~~ &  \sum_{i\in S} y_i\geq 0,& S\in \cB.\\
\end{aligned}
\label{eq:conv_prog}	
\end{equation}
By $g(e^y)$ in the above we mean $g(e^{y_1}, \ldots, e^{y_m})$. The objective is simply
$$y \mapsto \log \inparen{\sum_{\alpha} g_\alpha e^{\inangle{\alpha, y}}}.$$
Such a function (for nonnegative $g_\alpha$) is well known to be convex (which follows from H\"older's inequality). 

Note that Slater's condition for~\eqref{eq:conv_prog} is satisfied, hence strong duality holds. In order to derive the dual of the convex program~\eqref{eq:conv_prog} introduce multipliers $\mu_S\geq 0$ for every $S\in \cB$ and consider the Lagrangian
$$L(y,\mu) = \log g(e^y) -\sum_{S\in \cB} \mu_S \sum_{i\in S} y_i.$$
By taking the derivative with respect to $y_i$ and equating to zero, we obtain the following optimality condition:
$$e^{y_i}\cdot\frac{ \frac{\partial }{\partial z_i}g(z)\vert_{z=e^y}}{g(e^y)} =\sum_{S\ni i} \mu_S.$$
By summing up all these conditions for $i=1,2,\ldots, m$ we obtain $n$ on the left hand side (because $g$ is homogeneous) and $n\cdot \sum_{S\in \cB} \mu_S$ on the right. Hence, at optimality $\sum_{S\in \cB} \mu_S=1$. From strong duality, we obtain:
$$\max_{\substack{\mu\geq 0,\\\sum_{S\in \cB} \mu_S=1}}\min_{y\in \R^n}\inparen{ \log g(e^y) -\sum_{S\in \cB} \mu_S \sum_{i\in S} y_i }= \min_{\substack{ \sum_{i\in S}y_i \geq 0\\ \forall S\in \cB}} \log g(e^y) . $$
It remains to observe that $\sum_{S\in \cB} \mu_S \sum_{i\in S} y_i=\sum_{i=1}^m y_i \sum_{S\ni i}\mu_S$, hence what really matters are the marginals $\theta_i = \sum_{S\ni i} \mu_S$ and not the probability distribution $\mu$ itself. For this reason one can rewrite the above equality as:
$$\max_{\theta \in P(\cB)} \min_{y\in \R^n} \inparen{ \log g(e^y) -\sum_{i=1}^m y_i \theta_i}= \min_{\substack{ \sum_{i\in S}y_i \geq 0\\ \forall S\in \cB}} \log g(e^y).$$
The lemma follows by replacing $e^y$ by $z$ and taking exponentials on both sides.
\end{proof}

\subsection{Computability of $\cB-$Capacity}\label{ssec:computability}
In this section we prove that $\capa_\cB(g)$ can be efficiently evaluated whenever an evaluation oracle for $g$ and a separation oracle for $P(\cB)$ are provided. We apply the ellipsoid method to give a polynomial time algorithm. Before we give the proof, let us first establish an important fact, which will be useful later.

\begin{fact}\label{fact:gradient}
If $g\in \R[z_1, \ldots, z_m]$ is an $n-$homogeneous polynomial, and an evaluation oracle for $g$ is provided, then for every $i=1,2,\ldots, m$ and for every point $w \in \R^m$ the derivative $\frac{\partial}{\partial z_i} g(z)$ can be computed using at most $O(n)$ evaluations of $g$. 
\end{fact}
\begin{proof}
Without loss of generality consider the computation of $\frac{\partial}{\partial z_1} g(z)\vert_{z=w}$. Let us write $g(z)=\sum_{k=0}^n z_1^k r_k(z_2, \ldots, z_m)$, where $r_k$ is  a polynomial in the remaining variables $z_2, \ldots, z_m$. We need to compute $\sum_{k=1}^n k w_1^{k-1} r_k(w_2, \ldots, w_m)$. To this end, we can perform a univariate interpolation to calculate $r_k(w_2, \ldots, w_m)$ for every $k=0,1,\ldots, n$ and then just output the required result.
\end{proof}

\noindent When working with evaluation oracles, even if we do not have direct access to the coefficients, we need to know some upper bound on their description length, to state running time bounds. If the coefficients of $g$ are integers, such a bound can be evaluated algorithmically by just querying $g(1,1,\ldots, 1)$ (recall that we assume all the coefficients of $g$ to be nonnegative).

\begin{lemma}\label{lemma:computability}
If $g\in \R^+[z_1, \ldots, z_m]$ is an $n-$homogeneous polynomial given by an evaluation oracle and $\cB \subseteq {[m] \choose n}$ is a family of sets given by a separation oracle for $P(\cB)$, then $\capa_\cB(g)$ can be computed up to a multiplicative precision of $(1+\eps)$ in time polynomial in $m$, $\log \frac{1}{\eps}$ and $L$, where $L$ is an upper bound on the description size of coefficients of $g$.
\end{lemma}
\begin{proof}
We first rewrite $\capa_\cB(g)$ as a convex program using its dual characterization~\ref{lemma:dual_capa}.

\begin{equation}
\begin{aligned}
	\inf_{y\in \R^m}  ~~ &  \log g(e^y),&\\
	\st ~~ &  \sum_{i\in S} y_i\geq 0,& S\in \cB.\\
\end{aligned}
\label{eq:conv_capa}	
\end{equation}
Where by $g(e^y)$ we mean $g(e^{y_1}, \ldots, e^{y_m})$. 
Note that $f(y)=g(e^y)$ is a convex function of $y$, indeed this follows from H\"older's inequality, since all the coefficients of $g$ are nonnegative. Moreover the set $Y=\{y\in \R^m: \sum_{i\in S} y_i \geq 0, \mbox{ for all } S\in \cB\}$ is convex and a separation oracle for it can be constructed given the separation oracle for $P(\cB)$. 

By the standard binary search technique we reduce the problem to answering a sequence of queries of the form:
$$\mbox{Is the set }~~U_c=\{y\in Y: f(y)\leq c\} ~~\mbox{ nonempty?}$$
To solve it using ellipsoid we just need to provide a separation oracle. If the condition $y^0 \in Y$ is not satisfied then the separation oracle for $Y$ provides us with a separating hyperplane. If we are
given a point $y^0$ such that $f(y^0)=c_0>c$ we can produce a separating hyperplane using the gradient information, because from (first order) convexity we have:
$$f(y)\geq f(y^0) + \inangle{\nabla f(y^0), y-y^0} \geq c_0 + \inangle{\nabla f(y^0), y-y^0},$$
hence in particular the equation:
$$\inangle{\nabla f(y^0), y-y^0} = -\frac{c_0-c}{2},$$
gives a separating hyperplane.

\end{proof}

\subsection{The Generating Polynomial}\label{sec:selection1}
As demonstrated in Section~\ref{sec:capa} the task of proving that $\capa_\cB(g)$ well approximates $g_\cB$ boils down to coming up with a real stable polynomial $h$ which is a $\cB^\star-$selection and its lower capacity with respect to $\cB^\star$ is as large as possible. We will provide one generic way of coming up with such polynomials and proving lower bounds on their lower capacity. It captures the case when $\cB^\star$ is a strongly Rayleigh matroid. In the next subsection we extend it to capture more general settings. 

Recall that for a given family $\cB \subseteq {[m] \choose n}$ (which should be thought as the dual of what we normally call $\cB$) we are interested in an $n-$homogeneous real stable polynomial $h$, which is a $\cB-$selection and has a large lower $\cB$-capacity. There exists one  natural choice for $h$, the generating polynomial of $\cB$ 
$$h(z) = \sum_{S\in \cB}z^S,$$
it satisfies the conditions for a $\cB-$selection in an obvious way, thus the remaining questions are:
\begin{itemize}
\item Is $h(z)$ real stable?
\item What is the lower $\cB-$capacity of $h$?
\end{itemize}
The first question leads us directly to the notion of Rayleigh and strongly Rayleigh matroids introduced in \cite{CW06}. As shown in~\cite{Branden07} the class of matroids $\cM$ such that $h(z)$ is real stable (for $\cB$ being the set of bases of $\cM$) is precisely equal to the class of matroids enjoying the strongly Rayleigh property. In the next subsection we discuss possible ways to weaken this requirement of $\cB$ being strongly Rayleigh by manipulating the coefficients of $h(z)$.

Now we address the second question from the above list.
\begin{lemma}\label{lemma:dual_capa}
For every non-empty family $\cB \subseteq {[m] \choose n}$, the polynomial $f(z)=\sum_{S\in \cB} \beta_S z^S$, with $\beta_S\geq 1$ for every $S\in \cB$, satisfies
$$\dcapa_\cB(f) \geq 1.$$
\end{lemma}
\begin{proof}
We need to prove that for every choice of $\theta \in P(\cB)$ we have
$$\inf_{z>0} \frac{f(z)}{\prod_{i=1}^m z_i ^{\theta_i }}\geq 1.$$
Let us then fix $\theta \in P(\cB)$ and any $z>0$. Since $\theta \in P(\cB)$ we can write it as
$\theta=\sum_{S\in \cB} \alpha_S 1_S,$
for some nonnegative $\alpha \in [0,1]^\cB$ with $\sum_S \alpha_S = 1$. From Jensen's inequality we obtain:
\begin{align*}
\log \inparen{\sum_{S\in \cB} \alpha_S z^S} &\geq \sum_{S\in \cB} \alpha_S \log \inparen{z^S}\\
&=\sum_{S}\alpha_S \sum_{i\in S} \log z_i\\
&=\sum_{i=1}^m \sum_{S\ni i} \alpha_S \log z_i\\
&=\sum_{i=1}^m \theta_i \log z_i
\end{align*}
By taking exponentials and using the trivial estimate $f(z) \geq \sum_{S\in \cB}\alpha_S z^S$ we obtain the lemma.
\end{proof}
It is not hard to see that the above lemma is actually tight and indeed $\dcapa_\cB(f)$ is equal to $\min \{\beta_S: S\in \cB \}.$

\subsection{Nonuniform Generating Polynomials and Linear Matroids}\label{sec:selection2}
Consider the case when $\cB \subseteq {[m] \choose n}$ is the set of bases of a linear matroid $\cM$ (over $\R$). This means that there is a matrix $V\in \R^{m\times d}$ having rows $v_1, v_2, \ldots, v_m$ such that:
$$S\in \cM~~~~\Leftrightarrow~~\mbox{ the collection $\{v_i\}_{i \in S}$ is linearly independent.}$$
this can be also restated as:
$$S\in \cM~~~~\Leftrightarrow~~\det(V_S^\top V_S)>0,$$
where $V_S$ is the matrix $V$ restricted to rows with indices in $S$. Such a matrix $V$ we call a linear representation of $\cM$. Note that the polynomial:
$$h(z)=\sum_{S\in {[m] \choose n}} \det(V_S^\top V_S) z^S$$
has as its support exactly $\cB$. We have the following well known, but important fact. 
\begin{fact}\label{fact:linear_stable}
For any matrix $V\in \R^{m\times d}$ the polynomial $h(z)=\sum_{S\in {[m] \choose n}} \det(V_S^\top V_S) z^S$ is $n-$homogeneous  and real stable.
\end{fact}
\begin{proof}
We present the proof in the special case when $d=n$. We refer the reader to~\cite{NS16} for the general case.
Consider the polynomial $r(z)=\det\inparen{\sum_{i=1}^m z_i v_i v_i^\top}=\det\inparen{V^\top ZV}$, where $Z=\diag{z_1, \ldots, z_m}$. As a determinantal polynomial, $r(z)$ is real stable (see e.g.~\cite{Vishnoi-Survey}). It remains to observe that from the Cauchy Binet formula:
$$r(z)=\sum_{S\in {[m] \choose n}} z^S \det(V_S^\top V_S).$$
\end{proof}
This fact allows us to use $h(z)$, after a suitable rescaling, as an approximate $\cB-$selection. This rescaling can be controlled by the \textit{unbalance} of a linear matroid which essentially measures the maximum distortion $\frac{\det(V_S^\top V_S)}{\det(V_T^\top V_T)}$ over $S,T\in \cB$ (see definition in Section~\ref{sec:preliminaries}). We obtain

\begin{lemma}\label{lemma:linear_selec}
Let $\cB \subseteq {[m] \choose n} $ be a set of bases of an $\R-$linear matroid. There exists a real stable $c-$approximate $\cB-$selection $h$ with $c\leq \un(\cB)$ (the unbalance of $\cB$) and $\dcapa_{\cB}(h)\geq 1$.
\end{lemma}
\begin{proof}
Take $V$ to be the most balanced representation of $\cB$, i.e. $\un(V)=\un(\cB)$. We can scale the vectors of $V$ in such a way that:
$$1 \leq \det(V_S V_S^\top) \leq \un(V)~~\mbox{for all }S\in \cB.$$
From the Fact~\ref{fact:linear_stable} the polynomial:
$$h(z)=\sum_{S\in {[m] \choose n}} \det(V_S V_S^\top) z^S$$
is real stable, and clearly it is an $\un(V)-$approximate $\cB-$selection.
\end{proof}
Also, more generally we can state the following lemma on nonuniform generating polynomials.
\begin{lemma}\label{lemma:rayleigh_selec}
Let $\cB \subseteq {[m] \choose n} $ be a set of bases of a matroid. Suppose that $\mu$ is a strongly Rayleigh distribution with $\supp(\mu)=\cB$ and $P=\max_{S,T\in \cB} \frac{\mu(S)}{\mu(T)}$. Then there exists a real stable $P-$approximate $\cB-$selection $h$ with $\dcapa_{\cB}(h)\geq 1$.
\end{lemma}
\begin{proof}
Take $h(z)$ to be $\frac{1}{\min_{S\in \cB} \mu(S)} \sum_{T\in \cB} z^T \mu(T)$. We have that $1\leq h_S \leq P$ for every $S\in \cB$. Since $h\in \R_1^+ [z_1, \ldots, z_m]$ and $\supp(h)=\cB$, $h$ is a $P-$approximate $\cB-$selection. Also, $h(z)$ is real stable, because $ \sum_{T\in \cB} z^T \mu(T)$ is.
\end{proof}

\subsection{Uniform and Partition Matroids}\label{ssec:partition}
The generic choice of an $\cB-$selection $h(z)$ to be  $\sum_{S\in \cB} z^S$ is indeed natural and intuitively ``right'', however it seems to be suboptimal. Two important cases where we can  provably surpass this suboptimality are uniform matroids and partitions matroids. The result for uniform matroids follows from the work~\cite{AOSS16}, below we extend it to the case of partition matroids.

\begin{lemma}\label{lemma:partition_selec}
Let $\cB \subseteq {[m] \choose n} $ be a set of bases of a partition matroid, i.e. $\cB$ is of the form:
$$\cB=\{S: |S\cap P_j|=b_j \mbox{ for }j=1,2,\ldots, p\},$$
where $P_1, P_2, \ldots, P_p$ is a partition of $[m]$ into $p$ disjoint sets and $\sum_{j=1}^p b_j=n$. Then there exists a real stable $\cB-$selection $h(z)$ with
$$\dcapa_{\cB}(h)\geq \prod_{j=1}^p \frac{b_j^{b_j}}{b_j!}.$$
\end{lemma}

\begin{proof}
Consider the following choice of $h$
$$h(z) = \prod_{j=1}^p\frac{\inparen{\sum_{i\in P_j}z_i }^{b_j}}{b_j!}.$$
It is not hard to see that $h(z)$ is indeed a $\cB$-selection. Indeed, all the coefficients of monomials $z^S$ for $S\in \cB$ are $1$ and these are the only squarefree monomials which appear with a nonzero coefficient.

  It suffices to show a lower bound on $\dcapa_{\cB}$. To this end fix any $\theta \in P(\cB)$, any $z>0$ and consider
\begin{equation}\label{eq:product}
\frac{h(z)}{\prod_{i=1}^m z_i^{\theta_i}} = \prod_{j=1}^p\frac{\inparen{\sum_{i\in P_j}z_i }^{b_j}}{b_j!\prod_{i\in P_j}z_i^{\theta_i}}. 
\end{equation}
At this point one can observe that all the terms in the product can be analyzed separately, this is because $P(\cB)$ is actually a cartesian product of $p$ sets. More precisely
$$P(\cB) = \{\theta \in [0,1]^m: \sum_{i\in P_j} \theta_i = b_j \mbox{ for all }j=1,2,\ldots, p\}.$$

\noindent Let us consider one particular term in the product in~\eqref{eq:product}. WLOG take $j=1$, and assume that $P_1=\{1,2,\ldots, d\}$ and rename $b_1$ to simply $b$. Our goal now reduces to lower-bounding:
$$\frac{(\sum_{i=1}^d z_i)^b}{b! \prod_{i=1}^d z_i^{\theta_i}},$$
where $\theta\in [0,1]^d$ satisfies $\sum_{i=1}^d \theta_i=b$. From Jensen's inequality for $\log$ we obtain:
$$b \cdot \log \inparen{\sum_{i=1}^d z_i} = b \cdot \log \inparen{\sum_{i=1}^d \frac{\theta_i}{b} \cdot\frac{b z_i}{\theta_i}}\geq \sum_{i=1}^d \theta_i \log \frac{bz_i}{\theta_i}.$$
Further,
$$\sum_{i=1}^d \theta_i \log \frac{bz_i}{\theta_i} = \sum_{i=1}^d \theta_i \log \frac{b}{\theta_i} + \sum_{i=1}^d \theta_i \log z_i.$$
Finally note that $\inbraces{\frac{\theta_i}{b}}_{i\in [d]}$ is a probability distribution over $d$ items, with probabilities bounded from above by $\nfrac{1}{b}$, this implies that its negative entropy is at least
$$-\sum_{i=1}^d \frac{\theta_i}{b} \log \frac{\theta_i}{b} \geq \log b.$$
Concluding, we obtain:
$$b \cdot \log \inparen{\sum_{i=1}^d z_i} \geq \sum_{i=1}^d \theta_i \log z_i + b \log b.$$
After taking exponentials and dividing by $b!$
$$\frac{(\sum_{i=1}^d z_i)^b}{b! \prod_{i=1}^d z_i^{\theta_i}}\geq \frac{b^b}{b!}.$$
\end{proof}

\subsection{Proofs}\label{ssec:proofs}
We conclude our discussion in this section and provide complete proofs of Theorems~\ref{thm:finite} and~\ref{thm:bounds}.
\begin{proofof}{of Theorem~\ref{thm:finite} }
Since there is a strongly Rayleigh distribution supported on $\cB^\star$, applying Lemma~\ref{lemma:rayleigh_selec} we can conclude existence of a real stable $P-$approximate $\cB^\star-$selection for some (possibly large) $P>0$, such that $\dcapa_\cB(h)\geq 1$. Now, the approximate version of Lemma~\ref{lemma:approx} (see Remark~\ref{remark:approx_selec}), concludes the proof.
\end{proofof}

\begin{proofof}{of Theorem~\ref{thm:bounds} }
For points 1. and 2. we reason as in the proof of Theorem~\ref{thm:finite}, we construct suitable $\cB^\star$ selections using Lemmas~\ref{lemma:rayleigh_selec}  and \ref{lemma:linear_selec} respectively and then  apply the approximate variant of Lemma~\ref{lemma:approx}.  Note that the above mentioned $\cB^\star-$selections are multilinear, hence we can apply the sharper $2^{-m}$ bound in Lemma~\ref{lemma:approx} in the case when $g$ is multilinear as well. 

For point 3. the bound on $\mlb$ follows implicitly from a reasoning in~\cite{NS16}, we derive it in Section~\ref{app:partition}. To obtain a bound on $\mb$, note that $\cB^\star$ is a partition matroid $\{U(P_j, |P_j|-b_j)\}_{j\in [p]}$. We set $h(z)$ to be the $\cB^\star-$selection constructed in Lemma~\ref{lemma:partition_selec} and apply Lemma~\ref{lemma:approx}, this yields a bound  $\mb \leq \frac{m^m}{m!} \cdot \prod_{j=1}^p \frac{(|P_j|-b_j)!}{(|P_j|-b_j)^{|P_j|-b_j}}$. 

It remains to argue that 
$\frac{m^m}{m!} \cdot \prod_{j=1}^p \frac{(|P_j|-b_j)!}{(|P_j|-b_j)^{|P_j|-b_j}}\leq e^n \inparen{\frac{2\pi m}{p}}^{p/2}$. After applying the bound $\frac{k!}{k^k} \leq \sqrt{k} e^{-k+1}$ and using $\sum_{j=1}^p (|P_j|-b_j) = m-n$ we obtain:
$$ \frac{m^m}{m!} \cdot \prod_{j=1}^p \frac{(|P_j|-b_j)!}{(|P_j|-b_j)^{|P_j|-b_j}} \leq e^m e^{-m+n+p} \prod_{j=1}^p \sqrt{|P_j|-b_j}.$$
Finally, by the AM-GM inequality 
$$\prod_{j=1}^p \sqrt{|P_j|-b_j} \leq \inparen{\frac{m-n}{p}}^{p/2}\leq \inparen{\frac{m}{p}}^{p/2}.$$
\end{proofof}
\section{Optimization}\label{sec:optimization}

In this section we discuss the problem of finding 
\begin{equation}\label{eq:max}
\max_{S\in \cB} g_S
\end{equation} for a given $n-$homogeneous polynomial $g\in \R^+[x_1, \ldots, x_m] $ and a set family $\cB \subseteq {[m] \choose n}$.

One naive way to approach problem~\eqref{eq:max} would be to apply Theorem~\ref{thm:bounds} directly.  Since $\capa_\cB(g)$ approximates $g_\cB$ up to a factor of $\mb$ we could just output $\capa_{\cB}(g)$ as an approximation to~\eqref{eq:max} and obtain a guarantee of:
$${\mb} \cdot |\cB| \leq {\mb} \cdot m^n = {\mb} \cdot e^{n \log m}.$$
We believe that the bounds in Theorem~\ref{thm:bounds} can be strengthened, so that ${\mb}$ (or $\mlb$) depend on $n$ only (as for the case of uniform matroids), which makes the $e^{n \log m}$ inefficient. In the next subsection we propose a method which achieves an approximation guarantee of at most $\mb \cdot e^n$ and can be better, depending on a particular $\cB$.

\subsection{Convex Relaxation}\label{ssec:convex_relax}
  We consider the following relaxation to problem~\eqref{eq:max}.
\begin{equation}
\begin{aligned}
	\sup_{x\in P(\cB)}~~ \inf_{\substack{y>0\\\forall S\in \cB~y^S\geq 1 }} ~~ &  g(x_1y_1, \ldots, x_my_m).&\\
\end{aligned}
\label{eq:max_relax}	
\end{equation}
A similar relaxation was used in~\cite{NS16} in the context of subdeterminant maximization and in~\cite{AOSS16} for the Nash Social Welfare problem. In fact the relaxation in~\cite{NS16} is an upper bound for ours. Both of them behave similarly in the case when $\cB$ is a partition family. However, for other matroids, such as spanning tree matroids, the relaxation of~\cite{NS16} has an unbounded integrality gap, whereas for the case of relaxation~\eqref{eq:max_relax} we can prove
\begin{lemma}\label{lemma:approx_guarantee}
Let $\cB \subseteq {[m] \choose n}$ be a family of bases of a matroid. For every real stable, $n-$homogeneous polynomial $g\in \R^+[z_1, \ldots, z_m]$. The optimal value $OPT$ of the relaxation~\eqref{eq:max_relax} satisfies:
$$\frac{OPT}{{\mb} \cdot A(\cB)} \leq \max_{S\in \cB} g_S \leq OPT,$$
where $A(\cB):= \max \inbraces{\sum_{S\in \cB} x^S:x\in P(\cB)} \leq e^n$. Moreover, in the case when $g\in \R_1^+[z_1, \ldots, z_m]$, $\mb$ can be replaced by $\mlb$ in the bound above.
\end{lemma}
\begin{proof}
Let $p(x,y)$ denote the polynomial $g(x_1y_1, \ldots, x_my_m)$. One can easily see that if $x=1_S$ for some $S\in \cB$ then $p(x,y)\geq y^S \cdot g_S$, this implies that
$$\max_{S\in \cB} g_S \leq OPT.$$
To prove the lower-bound, fix the optimal solution $\bar{x}$ to the relaxation~\eqref{eq:max_relax} and consider the real stable polynomial $g(z)=p(\bar{x}, z)$. It follows that $OPT = \capa_\cB(g)$, hence:
$$OPT \leq {\mb} \cdot g_\cB.$$
We also have 
$$g_\cB = \sum_{S\in \cB}\bar{x}^S g_S \leq \inparen{\sum_{S\in \cB}\bar{x}^S} \cdot \max_{S\in \cB} g_S\leq A(\cB) \cdot OPT.$$
What remains to prove is that $A(\cB) \leq e^n$. This turns out to be a quite simple consequence of the constraints $x\geq 0$ and $\sum_{i=1}^n x_i=n$. Indeed, the largest possible value of the sum $\sum_{S\in \cB}\bar{x}^S$ is attained when $\bar{x}_i=\frac{n}{m}$ for all $i\in [m]$.
\end{proof}

\begin{lemma}\label{lemma:ellipsoid}
The relaxation~\eqref{eq:max_relax} is efficiently computable. Given access to an evaluation oracle for $g$ and a separation oracle for $P(\cB)$ one can obtain a $(1+\eps)-$approximation to the optimal solution of~\eqref{eq:max_relax} in time polynomial in $m$, $\log \frac{1}{\eps}$ and $L$, where $L$ is an upper bound on the description size of coefficients of $g$.\footnote{Note that the coefficients of $g$ are not explicitly given to us. For this algorithm to work, only the knowledge of $L$ is required.}
\end{lemma}
\begin{proof}
We first rewrite the relaxation~\eqref{eq:max_relax} in a convex form
\begin{equation}
\begin{aligned}
	\sup_{x\in P(\cB)}~~ \inf_{\substack{w\in \R^m \\\sum_{i\in S}w_i \geq 0, \forall S\in \cB}} ~~ &  \log g(x_1 e^{w_1}, \ldots, x_m e^{w_m})
\end{aligned}
\label{eq:NS_convex}	
\end{equation}
Denote $f(x,w)= \log g(x_1 e^{w_1}, \ldots, x_m e^{w_m})$. The function $f(x,w)$ is concave in the first variable $x$ (follows from real stability of $g$) and convex in the second variable $w$ (can be easily seen using H\"older's inequality). The constraints on both $x$ and $w$ are linear, indeed we have that $w \in W$, where
$$W=\{w\in \R^m: \sum_{i\in S}w_i \geq 0, \mbox{ for all } S\in \cB\}.$$
It is not hard to see that given a separation oracle for $P(\cB)$ one can construct one for $W$, hence our problem is in the general form of a concave-convex saddle point problem $\inf_{x\in X} \sup_{y\in Y} f(x,y)$ and we have separation oracles for both $X$ and $Y$. Below we sketch how such problems can be efficiently solved using the ellipsoid method, under the assumption that the gradient of $f$ can be efficiently computed. After that, we demonstrate that this is indeed the case for $f$.

By performing a binary search over the optimal value of the solution one can reduce the problem to solving a sequence of queries:
$$\mbox{Is the set }~~U_c=\{x\geq 0: x\in P(\cB),  \inf_{w\in W}~ f(x,w)\geq c\} ~~\mbox{ nonempty?}$$
Note that $U_c$ is convex, so we can use the ellipsoid algorithm to test emptiness (approximately). It is enough to provide a separation oracle for $U_c$, which in turn reduces to finding a separation oracle for: $S_c=\{x \in \R^m: \inf_{w\in W}~ f(x,w)\geq c\}$. We focus on this task now.

Let $x\in  \R^m$ be given, we would like to answer the question if $x\in S_c$ and if not, provide a separating hyperplane. To this end we solve the optimization problem $\min_{w\in W} f(x,w)$ (see Lemma~\ref{lemma:computability}) if the optimal value is say $c_0<c$ and is attained at some point $w^\star$, we know (from concavity in the first variable) that:
$$f(z,w^\star) \leq f(x,w^\star) + \inangle{\nabla_x f(x,w^\star), z-x} \leq c_0 + \inangle{\nabla_x f(x,w^\star), z-x}$$
hence 
$$\inangle{\nabla_x f(x,w^\star), z-x} = \frac{c-c_0}{2}$$
is a separating hyperplane (for fixed $x$ and varying $z$).

Finally let us observe that we can compute the gradient of $f$ given just an evaluation oracle of $g$. To see this, observe that computing $\frac{\partial}{\partial x_i} f(x,w)$ and $\frac{\partial}{\partial w_i} f(x,w)$ reduces to computing $\frac{\partial}{\partial z_i} g(z)$, which is efficiently computable by Fact~\ref{fact:gradient}. 
\end{proof}

\noindent The results of this section allow us to deduce Theorem~\ref{thm:max}.
\begin{proofof}{of Theorem~\ref{thm:max}}
Given an optimization problem~\eqref{eq:max} we apply the relaxation~\eqref{eq:max_relax} to it. Lemma~\ref{lemma:ellipsoid}  guarantees that it can be solved in polynomial time. Now by applying Lemma~\ref{lemma:approx_guarantee} we obtain the claimed approximation guarantee.
\end{proofof}

\subsection{Application to Maximizing Subdeterminants}\label{ssec:maxdet}
In this section we discuss the problem of maximizing subdeterminants under constraints. Let $L\in \R^{m\times m}$ be a symmetric PSD matrix  and let $\cB \subseteq {[m] \choose n}$ be any family of sets. We consider the problem:
\begin{equation}\label{eq:max_det}
\max_{S\in \cB} \det(L_{S,S}),
\end{equation}
where $L_{S,S}$ denotes the submatrix of $L$ obtained by restricting it to rows and columns from $S$. Typically in this setting it is useful to consider the Cholesky decomposition $L=VV^\top$ for some matrix $V\in \R^{m\times d}$. This allows us to write the generating polynomial:
$$g(z) = \sum_{S\in {[m] \choose n}} z^S \det(L_{S,S}) =  \sum_{S\in {[m] \choose n}} z^S \det(V_S^\top V_S)=\det(V^\top Z V)$$
where $Z=\diag{z_1, \ldots, z_m}$. We can then apply the result of Theorem~\ref{thm:max} to obtain Corollary~\ref{cor:subdet}.

\begin{proofof}{of Corollary~\ref{cor:subdet}}
We consider $g(z)=\sum_{S\in {[m] \choose n}} z^S \det(L_{S,S}) $ as above. As observed in Fact~\ref{fact:linear_stable} $g(z)$ is $n-$homogeneous and real stable. Moreover $g(z)$ is efficiently computable, since it just boils down to computing a determinant of a $d\times d$ matrix. Through the optimization problem~\eqref{eq:max} $g(z)$ encodes exactly~\eqref{eq:max_det}, hence we can apply Theorem~\ref{thm:max}.
\end{proofof}

\bibliographystyle{alpha}
\bibliography{references}

\begin{appendix}
\section{Entropy Interpretation of $\cB-$capacity}\label{sec:entropy}
In this section we show that computing $\cB-$capacity of a polynomial $p$ can be equivalently seen as finding a distribution $q$ minimizing the KL-divergence between $p$ and $q$ subject to marginal constraints. In this context it is convenient to think of $p\in \R^+[z_1, \ldots, z_m]$ as a probability distribution over monomials $z^\alpha$, more precisely, the probability of a monomial $z^\alpha$ is simply $\frac{p_\alpha}{p(1)}$. (Note that $p(1)=\sum_{\alpha} p_\alpha$. ) We will use $\Lambda_n \defeq \{\alpha \in \N^m : |\alpha|=n\},$ to denote all monomials of degree $n$. 

\begin{lemma}\label{lemma:entropy}
Given a real stable, $n-$homogeneous polynomial $p\in \R^+[z_1, \ldots, z_m]$ and $\cB \subseteq {[m] \choose n}$, the capacity  $\capa_\cB(p)$ is equal to the exponential of the optimum value of the following convex optimization problem
\begin{equation}
\begin{aligned}
	\sup_{q,\theta}  ~~ &  -KL(q,p),&\\
		\st ~~ &  \sum_{\alpha \in \Lambda_n} q_\alpha=1,\\
		& \sum_{\alpha \in \Lambda_n} q_\alpha \cdot \alpha = \theta, \\
 &  \theta\in P(\cB),\\
	&  q\geq 0.&
\end{aligned}
\label{eq:entropy}	
\end{equation}
where $q$ is a vector indexed by all possible multi-indices $\alpha \in \N^m$ with $|\alpha|=n$ and $KL(q,p) \defeq \sum_{\alpha \in \N^m} q_\alpha \log \frac{q_\alpha}{p_\alpha}$.
\end{lemma}

\noindent Before we start proving Lemma~\ref{lemma:entropy}. Let us try to interpret what it means. As already mentioned, $p$ can seen as a probability distribution over monomials, or in other words over multisubsets of $[m]$ of cardinality $n$, which we represent by $\Lambda_n$.  If $q$ is a distribution over $\Lambda_n$, then $\theta = \sum_{\alpha \in \Lambda_n} q_\alpha \cdot \alpha$ is the marginal vector of $q$ in which case $\theta_i$ (for $i\in [m]$) represents the expected number of copies of $i$ in a sample $\alpha$ drawn according to $q$. 

The optimization problem~\eqref{eq:entropy} asks to find a distribution $q$ over $\Lambda_n$ which is the closest (in relative entropy) to the given distribution $p$ under the constraint that its marginal lies in $P(\cB)$. 

\begin{proof}
The proof relies on convex duality. Fix any probability distribution $p$ on $\Lambda_n$ and a marginal vector $\theta \in P(\cB)$. We derive the dual of the following convex program
\begin{equation}
\begin{aligned}
	\max_{q,\theta}  ~~ &  -KL(q,p),&\\
		\st ~~ &  \sum_{\alpha \in \Lambda_n} q_\alpha=1,\\
		& \sum_{\alpha \in \Lambda_n} q_\alpha \cdot \alpha = \theta, \\
	&  q\geq 0.&
\end{aligned}
\label{eq:entropy_theta}	
\end{equation}
We make a simplifying assumption that there exists a $q>0$ such that $\theta= \sum_{\alpha \in \Lambda_n} q_\alpha \alpha$. This implies that the Slater's condition is satisfied for~\eqref{eq:entropy_theta}, which makes the analysis much simpler. Introduce Lagrangian multipliers $z\in \R$ and $\lambda \in \R^m$  and consider the Lagrangian function:
$$L(q,\lambda,z) = -KL(q,p) - \lambda^\top \inparen{\sum_{\alpha \in \Lambda_n}q_\alpha \cdot \alpha - \theta}-z\cdot \inparen{\sum_{\alpha \in \Lambda_n} q_\alpha - 1}$$
We are going to derive a formula for $g(\lambda, z) = \max_{q} L(q,\lambda, z)$. To this end, we derive optimality conditions with respect to $q$.
$$\frac{\partial}{\partial q_\alpha} L = -\log q_\alpha -1+ \log p_\alpha - \lambda^\top \alpha - z=0.$$
Note that the above implies that $q>0$ and this is why we do not need to introduce dual variables for non-negativity constraints. Using the above conditions we obtain
$$g(\lambda , z) =\sum_{\alpha \in \Lambda_n}p_\alpha e^{-\lambda^\top \alpha - z-1}+\lambda^\top \theta +z .$$
Hence the dual can be written as:
$$
	\min_{\lambda\in \R^m , z\in \R}  ~~  \sum_{\alpha \in \Lambda_n} p_\alpha e^{-\lambda^\top \alpha - z-1}+\lambda^\top \theta +z.
$$
We eliminate the $z$ variable from the above by minimizing the objective with respect to $z$. Thus we obtain:
\begin{equation}\label{eq:dual2}
	\min_{\lambda\in \R^m}  ~~\log \inparen{ \sum_{\alpha \in \Lambda_n} p_\alpha e^{-\lambda^\top \alpha }}+\lambda^\top \theta .
\end{equation}
Because of our assumption, Slater's condition is satisfied and hence strong duality holds, thus we obtain equality between the optimal value of~\eqref{eq:entropy_theta} and~\eqref{eq:dual2}.
To obtain the desired form, replace $-\lambda_i\in \R$ by $\log z_i$ for $z>0$. This gives
\begin{equation}\label{eq:dual2}
	\min_{z>0}  ~~\log \inparen{ \sum_{\alpha \in \Lambda_n} p_\alpha z^ \alpha }-\sum_{i=1}^m \theta_i \ln z_i .
\end{equation}
and after taking the exponential we recover the familiar
\begin{equation}\label{eq:dualcapa}
\min_{z>0} \frac{p(z)}{\prod_{i=1}^m z_i^{\theta_i}}.
\end{equation}
Going back to our assumption that $\theta$ can be obtained as $\sum_{\alpha} q_\alpha \alpha$ with $q>0$. After dropping this assumption, the equality we established above still holds, but several complications arise, one of them being  that the minimum value might not be attained (only in the limit). We skip the proof in the general case. 

It is now enough to observe that taking the maximum over $\theta \in P(\cB)$ of~\eqref{eq:dualcapa} gives $\capa_\cB(p)$, hence indeed we obtain equality between $\capa_\cB(p)$ and the exponential of the optimal value of~\eqref{eq:entropy_theta}.

\end{proof}
\noindent 
Assume now for brevity that $p$ is normalized, i.e. $p(1)=1$. In case when its marginal vector $\theta$ already lies in $P(\cB)$, we know that (since $KL(q,p)\geq 0$), the optimal solution to~\eqref{eq:entropy} is $q=p$ and hence $\capa_\cB(p)=1$. In view of our results, this implies a quite surprising fact. If $p$ and $\cB$ satisfy the assumptions of Theorem~\ref{thm:finite} then we can lower-bound $p_\cB=\sum_{S\in \cB} p_S$ by an absolute positive number (not depending on $p$). If $p$ was an arbitrary polynomial then we could easily make its marginal lie in $P(\cB)$ (and thus $\capa_\cB(p)=1$) while keeping $p_S=0$ for all $S\in \cB$.

\section{Counterexamples}\label{sec:examples}
We provide examples showing that $\capa_\cB(g)$ might not give a good approximation to $g_\cB$ if either $\cB$ is not a family of bases of a matroid or $g$ is not real stable.

\parag{Example 1.} We consider the case when $m=4$ and $n=2$. Consider a polynomial $g(z) = (z_1+z_2)(z_3+z_4)$ which is real stable, as a product of real stable polynomials. We pick a family $\cB$ which is not a set of bases of a matroid, namely $\cB=\{\{1,2\}, \{3,4\}\}$. Using the dual characterization of $\capa_\cB$ we have:
$$\capa_\cB(g) = \inf_{\substack{ z>0\\ z_1z_2\geq 1, z_3z_4\geq 1}} g(z).$$
Hence we obtain that $\capa_\cB(g)=4$, while clearly $g_\cB = 0$. 

\parag{Example 2.} Consider now a similar example, with the roles of $g$ and $\cB$ ``reversed''. Let $\cB=\{\{1,3\}, \{1,4\}, \{2,3\}, \{2,4\}\}$ be a family of bases of a partition matroid and $g(z)=z_1 z_2 + z_3 z_4$. Again we have:
$$\capa_\cB(g)\geq 2,$$
while $g_\cB=0$. In this example the polynomial $g(z)$ is not real stable. Indeed if $\omega= e^{\frac{2\pi i}{8}}$ then $(\omega, \omega, \omega^3, \omega^3)$ is a root of $g$ with positive imaginary part. 

\section{Lower bound on $\mb$}\label{app:lower_bound}
We prove a lower bound on $\mb$, which can be seen as another interpretation of the fact that the quality of approximating the permanent of a nonnegative matrix $A$ by the capacity of its product polynomial $p_A$ can be $e^n$ in the worst case.
\begin{fact}
There exists a family $\cB \subseteq {[m] \choose n}$ of bases of a partition matroid  for which $\mb \geq e^{\sqrt{m}}$.
\end{fact}
\begin{proof}
Consider a universe $U=[n]\times [n]$ of cardinality $m=n^2$. Bases of the considered matroid have exactly one element in every part $\{i\} \times [n]$ for $i=1,2,\ldots, n$. Formally:
$$\cB = \{\{(1,f(1)), (2,f(2)), \ldots, (n,f(n))\}: f\in [n]^{[n]}\}.$$
Of course $|\cB|=n^n$. Consider now a polynomial $g(z) = \inparen{\sum_{i,j} z_{i,j}}^n$, clearly $g(z)$ is real stable and it can be proved that $\capa_\cB(g) \geq  n^{2n}$. Indeed, from the AM-GM inequality:
$$g(z) =\inparen{\sum_{i,j} z_{i,j}}^n \geq( n^2)^n \prod_{i,j} z_{i,j}.$$
Under the condition that $z^S \geq 1$ for every $S\in \cB$ the above implies that $g(z)\geq n^{2n}$. 

It is also easy to calculate that $g_\cB = |\cB| \cdot n!=n^n \cdot n!$, hence we obtain $\mb \geq \frac{\capa_\cB(g)}{g_\cB} =\frac{n^{2n}}{n^n \cdot n!} \approx e^n=e^{\sqrt{m}}.$
\end{proof}

\section{A Bound for Partition Matroids}\label{app:partition}
The following lemma appeared implicitly in~\cite{NS16} in the context of determinantal polynomials.
\begin{lemma}
Let $\cM$ be a partition matroid $\{U(P_j, b_j)\}_{j\in [p]}$ then $\mlb \leq \frac{n^n}{n!} \prod_{j=1}^p \frac{b_j!}{b_j^{b_j}}$.
\end{lemma}
\begin{proof}
Let $n=\sum_{j=1}^p b_j$ be the rank of the partition matroid. Consider an $n-$homogeneous polynomial $g\in \R_1^+[z_1, \ldots, z_m]$. We perform the following \textit{symmetrization} procedure. For every part $P_j$ we introduce $b_j$ new variables $u_{j,1}, u_{j,2}, \ldots, u_{j,b_j}$ and define $s_j = \sum_{i=1}^{b_j} u_{j,i}$. For notational convenience, let us define a function $\sigma: [m] \to [p]$ which given an element $e\in [m]$ outputs an index $\sigma(e)$ such that $e\in P_{\sigma(e)}$. We consider
$$f(u) = g(s_{\sigma(1)}, s_{\sigma(2)}, \ldots, s_{\sigma(m)}).$$ 
Note that now $f(u)$ is $n-$variate, $n-$homogeneous and real stable. Moreover, we can relate $g_\cB$ to the coefficient (call it $c$) of $\prod_{i,j} u_{i,j}$ in $f(u)$ as follows:
$$c=p_\cB \cdot \prod_{j=1}^p b_j!.$$
By Gurvits' inequality for $n-$variate $n-$homogeneous polynomials we have
$$\capa(f) \leq \frac{n^n}{n!} c.$$
We will use the following simple upper bound on $\capa_\cB(g)$
$$\capa_\cB(g) = \inf_{\substack{ z>0\\ \forall S\in \cB~z^S\geq 1}}g(z)\leq \inf_{\substack{ z>0\\ \forall S\in \cB~z^S= 1}}g(z).$$
Note that importantly
$$\inf_{\substack{ z>0\\ z^S= 1, S\in \cB}}g(z) = \frac{1}{\prod_{j=1}^p b_j^{b_j}} \cdot \inf_{\substack{ u>0\\ \prod_{i,j} u_{i,j}=1}}f(u) =  \frac{1}{\prod_{j=1}^p b_j^{b_j}} \cdot \capa(f).$$
this equality follows because the constraints $z^S=1$ for every $S\in \cB$, imply  that the value of $z_i$ is constant inside every partition $P_j$. Hence there exists a one-to-one mapping between feasible $z$ and feasible $u$. As a consequence we obtain:
$$\capa_\cB(g) \leq \frac{1}{\prod_{j=1}^p b_j^{b_j}} \cdot \capa(f) \leq \frac{1}{\prod_{j=1}^p b_j^{b_j}}  \frac{n^n}{n!} \cdot c= p_\cB \cdot \frac{n^n}{n!} \cdot \prod_{j=1}^p \frac{b_j!}{b_j^{b_j}}.$$
\end{proof}

\end{appendix}

\end{document}